\documentclass[sigconf,natbib=true,screen,review=false, bookmarks=false]{acmart}
\usepackage{amsmath}
\usepackage{graphicx}
\usepackage[english]{cleveref}
\usepackage{amsfonts}
\usepackage[utf8]{inputenc}
\usepackage{mathtools}

\usepackage{xargs}
\usepackage{xcolor}
\usepackage{tikz}
\usepackage{ifthen}

\usepackage[inline]{enumitem}

\newcommand{\ie}{\emph{i.e.}}
\newcommand{\eg}{\emph{e.g.}}

\newenvironment{paragraphlist}
	{\begin{list}{}{}}
	{\end{list}}

\copyrightyear{2020}
\acmYear{2020}
\setcopyright{acmcopyright}\acmConference[ICSEW'20]{IEEE/ACM 42nd International Conference on Software Engineering Workshops }{May 23--29, 2020}{Seoul, Republic of Korea}
\acmDOI{10.1145/3387940.3391472}
\acmISBN{978-1-4503-7963-2/20/05}

\begin{document}
%
\title[Quantum-Annealing Based Software Components: A Case Study]{Quantum Annealing-Based Software Components:\\ An Experimental Case Study with SAT Solving
}

\author{Tom Krüger}
\affiliation{Ulm University\\
Germany}
\email{tom.krueger@uni-ulm.de}

\author{Wolfgang Mauerer}
\affiliation{Technical University of Applied Sciences Regensburg\\
Siemens AG, Corporate Research, Munich\\
wolfgang.mauerer@othr.de}


%
\begin{abstract}
Quantum computers have the potential of solving problems more efficiently than 
classical computers. While first commercial prototypes have become available,
the performance of such machines in practical application is still subject to
exploration. Quantum computers will not entirely replace classical machines,
but serve as accelerators for specific problems. This necessitates integrating quantum
computational primitives into existing applications.

In this paper, we perform a case study on how to augment existing
software with quantum computational primitives for the Boolean satisfiability 
problem (SAT) implemented using a quantum annealer (QA).  We discuss relevant quality
measures for quantum components, and show that mathematically
equivalent, but structurally different ways of transforming SAT to a QA can lead to
substantial differences regarding these qualities. We argue that engineers need to be aware
that (and which) details, although they may be less relevant in traditional software engineering, 
require considerable attention in quantum computing.
\end{abstract}

\maketitle


\section{Introduction}\label{sec:introduction}
%
The upcoming end of Moore's law and the trend towards energy efficient systems,
but the likewise ever-growing need for more computational power pose substantial challenges to
systems engineering and software architecture. New computational approaches that
substantially diverge from technologies established during the last decades start to graduate from
research laboratories into first working prototypes. Especially quantum computing has gained
substantial attraction during the last years~\citep{Peper2017TheEO}. Programming quantum
computers (QC) differs drastically from previously established techniques and approaches. 
Integrating QC into existing appliances must not only be addressed at the level of
algorithmic implementation, but also concerns many of the broader issues investigated 
in software engineering~\cite{Bass_2012}.  In this paper, we argue that
the problem at the current stage of development must be considered at a much lower level of
abstraction than is customary in software engineering, and illustrate this by a case
study of how to transition a core computational primitive---solving binary satisfiability
problems---from classical to quantum in existing software architectures. Our study 
illustrates that defining and testing specific quality properties of QC 
components is one of the crucial challenges. These properties do not play a central role in
traditional engineering, but must be considered in software architectures with quantum
components. We illustrate this by analysing
different approaches---one of which has been specifically designed for this paper, and improves
considerably on the state-of-the-art---to solving the binary satisfiability (SAT) problem.
We hope this helps readers to form a realistic intuition of near- and mid-term
capabilities, potentials and challenges of augmenting 
software with quantum components.




\section{Quantum Annealing}
By utilising quantum mechanical properties, QCs are expected to solve some problems 
more efficiently than their classical counterparts. Simulations of quantum systems~\cite{ortiz2001quantum} and
chemical reactions~\cite{reiher2017elucidating}, breaking of cryptographic codes~\cite{Shor1994PolynomialTimeAF}, but also 
optimising portfolios~\cite{venturelli2019reverse} are among the list of candidate problems.
Recent advances---although not undisputed---claim \emph{quantum supremacy}~\cite{arute2019quantum}, even if 
for extremely artificial problems.
Real-world adoption of quantum computing, as it matters to software engineering, is likely to happen
in an evolutionary way than by disruptive revolution.

We base our considerations on quantum \emph{annealers}: Many early potential industrial
use-cases~\cite{stollenwerk2019flight, neukart2017traffic, stollenwerk2019quantum} rest on this
class of machines, in part because they  were among the first offerings available for commercial
use (discussions about the full quantum mechanical nature of such machines~\cite{shin2014quantum}
are not relevant for our purposes).

Especially NP-complete problems, which are known to be classically intractable for inputs of growing
size when non-approximate solutions are desired, are candidates for which (polynomial) quantum speed-ups would be desirable.
Many NP-complete problems of practical interest are known. Especially the Boolean satisfiability problem (SAT)
has received substantial attention because many use cases, from system verification to constrained planning problems~\cite{karp1972reducibility}, have SAT at their core.
Quantum annealers are particularly well suited to process problems of this type~\cite{lucas2014ising}. They differ considerably from gate-based approaches in their physical realisation, and in the ways programs are engineered.

\subsection{Using Quantum Annealing Primitives in Existing Software Architectures}
Software engineering is (ignoring many aspects that
we cannot address for the lack of space) concerned with development, integration,
and testing (verification, validation, performance, quality, \dots) of
software~\cite{Bass_2012}.  This impacts quantum software development:

\subsubsection{Development}
A considerable body of previous research devoted to developing languages for programming
quantum systems focuses on gate-based approaches (\eg,~\cite{svore2006layered, haner2018software, steiger2018projectq}).  A growing number of quantum programming languages has been devised for this
hardware class (\eg,~\cite{mauerer2005semantics, green2013quipper, svore2018q}).  Roughly speaking, gate based quantum computers relate to quantum annealers like imperative programming languages do to declarative approaches. A deep understanding of quantum mechanics is not required to use current QA hardware, which is beneficial from a software point of view. Engineers can resort to techniques known from constraint programming, optimisation, and problem reduction. 

\subsubsection{Integration}
As Knill~\cite{knill1996conventions} discussed as early as 1996, quantum computers will not
entirely replace traditional machines, but will be part of hybrid quantum-classical
architectures, not unlike GPUs~\citep{S_Humble_2014, abbott2019hybrid} or other accelerators
(TPUs, FPGAs, \dots). Quantum annealers can be seen as hardware accelerators
for approximating quadratic unconstrained binary optimisation (QUBO)
problems.


The ability to easily replace functional components of a software architecture is a crucial
element of component-based software engineering~\cite{tahir2016framework,vale2016twenty}, and many
existing applications are designed along these ideas. In the following, we consider that a SAT solving
component is supposed to be replaced by a QA device in an existing software architecture. 

Replacing a library function call to a traditional solver by a network-based job submission
to a QA device is an easy programming tasks that we do not consider any further. However, two
data conversions are necessary, as Figure~\ref{fig:conversion} illustrates: The propositional calculus
formula for which a solution is sought must be mapped to a QUBO. Once the result of the optimisation
process is available, is must be translated back to the original SAT model. Both steps can be trivially
abstracted by an interface.

\begin{figure}[htb]
\includegraphics[width=\linewidth]{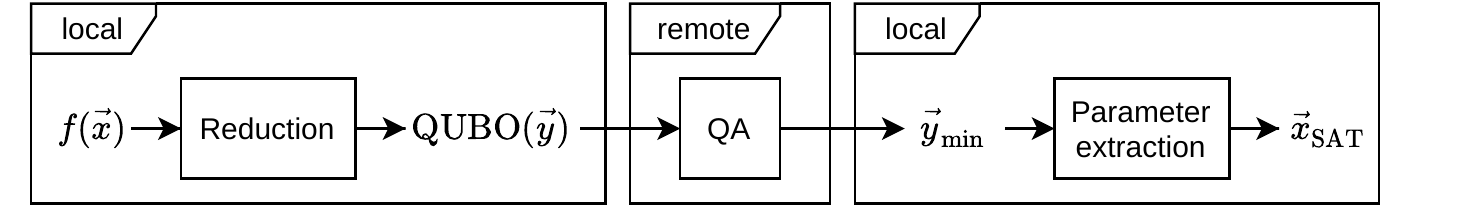}
\caption{Interface wrapping: Classical SAT solvers can be replaced with a QA based
implementation with limited effort.}\label{fig:conversion}
\end{figure}

\subsubsection{Testing}
Miranskyy and Zhang~\cite{Miranskyy_2018} discuss testing aspects related to verification and
validation of quantum programs. Fundamental properties of NP-complete problems guarantee that
solutions can be verified in polynomial time~\cite{Sipser:2006},  and consequently, validation and
verification of QA programs is not a core challenge.

However, quantum annealers usually only deliver \emph{approximate} solutions to problems,
and the quality of approximation is closely related to how ``programs'' (in the form
of mathematical reductions) are created. We focus on the issues arising from this
scenario in the rest of the paper.

\subsection{Workflow}\label{sec:workflow}
The workflow for solving problems on quantum annealers is more involved than for classical constraint optimisation. The necessary process comprises five stages, and choices in some of the stages
can greatly influence performance and accuracy of computations. Consequently, some knowledge of the inner
working of the quantum annealing process are useful. An AQO computation proceeds along the following stages~\citep{S_Humble_2014, McGeoch2014AdiabaticQC}:

\begin{paragraphlist}
\item[Problem Reduction] Like classical constraint optimization solvers, QA machines can optimise a specific class of models. Annealers can find solutions to \emph{quadratic unconstrained binary optimization} (QUBO)
  problems~\cite{Lewis2017QuadraticUB}, which are given by
  \begin{equation}
    \min[\vec{x}]\left(\sum_{i} c_{ii} x_i  + \sum_{i, j} c_{ij} x_i x_j\right)\label{eq:qubo}
  \end{equation}
  with \(x \in \{0,1\}\) and \(c \in \mathbb{R}\). A QUBO can  be represented by a weighted graph with nodes \(x_i\) and associated weights \(c_{ii}\). Weighted edges are given by \(c_{ij}\).

  Reducing a given problem \(p\) to a QUBO, \(p \leq \text{QUBO}\), requires no knowledge of quantum
  mechanics, and is similar to well-known reductions to Boolean satisfiability problems.
  As we will discuss later, structurally different (but logically equivalent) reductions can lead to drastically different
  performance on contemporary hardware. 
	
\item[Hardware Embedding] Software solvers can react dynamically to input, and easily build arbitrary data structures. For QA, the ``data structure'' used to represent a given input is fixed in hardware.
  This step ``translates'' an input onto the hardware structure~\cite{abbott2019hybrid, Choi2008} (see Figure~\ref{fig:embedding}). Mathematically equivalent reductions can lead
  to pronounced differences in solution quality, as we show in Section~\ref{sec:ksat_on_dwave}.
	
\item[Hardware Programming] The problem embedding needs to be transferred to the machine. The physical details
  of this operation are irrelevant to programmers, except that some parameters---most importantly, the
  duration of the annealing process---can be influenced. Finding an optimal value is currently a matter
  of experimentation.
\item[Execution] The machine finds a solution to the optimisation problem by ``executing'' a physical process.
\item[Post Processing]  Results obtained in the previous step are usually only \emph{close} to the desired optimum. Classical post-processing can improve solution quality~\citep{gabor2019assessing}. We will ignore
  this step in this paper since we are interested in the capabilities of QA as such,
  and not of classical data processing.
\end{paragraphlist}

\subsection{Experimentation Platform}\label{sec:dwave_platform}
All experiments that we discuss in the following were performed on a D-Wave 2000Q quantum annealer, model DW\_2000Q\_2\_1. The machine can be remotely accessed via a Python-based API. Performing computations requires
to specify a problem QUBO, and (essentially) anneal time and desired sample size \(n\). Once the QA has
evaluated the problem, a result set with \(n\) samples is returned. Each sample contains an assignment
for all qubits. 

 While it is possible to arbitrarily weigh the interaction between
qubits as specified by term \(c_{ij}\) from Eq.~(\ref{eq:qubo}), there are substantial restrictions on which
qubit \(i\) can \emph{physically} and \emph{directly} interact with which qubits \(j\) (see Ref.~\cite{cai2014practical} for details
on the available hardware graph structure). This limited connectivity poses a major practical challenge when
mapping logical to physical problems, since a pair of nodes that requires a logical connection
(a non-zero entry \(c_{ij}\) in Eq.~(\ref{eq:qubo})) must be represented by a chain of multiple
nodes on the hardware graph. This considerably limits the number of \emph{effectively usable} qubits as
compared to the number of \emph{physically available} qubits---Figure~\ref{fig:embedding_quality} exemplifies the problem visually. In general, longer chains lead to more undesirable physical perturbance,
and decrease result quality~\cite{Sax2020}. As a rule of thumb, the number of
usable logical qubits is only about 5-10\% the number of physical qubits.

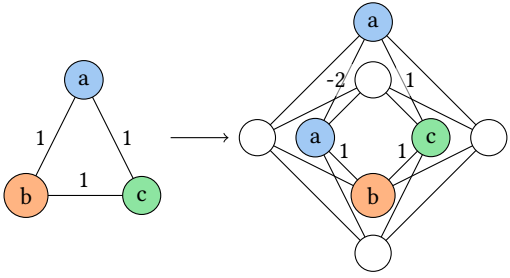
\begin{figure}
	\resizebox{0.8\linewidth}{!}{\begin{tikzpicture}[scale=1.5]
\definecolor{c_a}{rgb}{0.6313725490196078,
					   0.788235294117647,
					   0.9568627450980393};
					   
\definecolor{c_b}{rgb}{1.0,
					   0.7058823529411765,
					   0.5098039215686274};
					   
\definecolor{c_c}{rgb}{0.5529411764705883,
					   0.8980392156862745,
					   0.6313725490196078};
\colorlet{bg}{white}

\draw (0,0) node[circle, draw, fill=c_a]{a} 
	  -- node[anchor=east] {1}
	  ++(-0.5, -1) node[circle, draw, fill=c_b]{b}
	  -- node[anchor=south]{1}
	  ++(1, 0) node[circle, draw, fill=c_c]{c}
	  -- node[anchor=west]{1}
	  (0, 0);

\draw (2.5, 0.5)
	  -- node[fill=bg, opacity=.5, text opacity=1, xshift=3]{1}
	  ++(0.5, -1)
	  --
	  ++(-0.5, -1)
	  --
	  ++(-0.5, 1)
	  -- node[fill=bg, opacity=.5, text opacity=1, xshift=-3]{-2}
	  ++(0.5, 1);
	  
\draw (1.5, -0.5) node[circle, draw, fill=bg]{\:\:\:}
      --
      ++(1, 1) node[circle, draw, fill=c_a]{a}
      --
      ++(1, -1) node[circle, draw, fill=bg]{\:\:\:}
      --
      ++(-1, -1) node[circle, draw, fill=bg]{\:\:\:}
      --
      ++(-1, 1)
      --
      ++(1, 0.5) node[circle, draw, fill=bg]{\:\:\:}
      --
      ++(0.5, -0.5) node[circle, draw, fill=c_c]{c}
      -- node[anchor=south]{1}
      ++(-0.5, -0.5) node[circle, draw, fill=c_b]{b}
      -- node[anchor=south]{1}
      ++(-0.5, 0.5) node[circle, draw, fill=c_a]{a}
      --
      ++(0.5, 0.5)
      --
      ++(1, -0.5)
      --
      ++(-1, -0.5)
      --
      ++(-1, 0.5);
	  
\draw[->] (0.75, -0.5) -- (1.25, -0.5);
\end{tikzpicture}}
	\caption{Example for embedding a \emph{logical} graph that describes couplings between qubits
          (left) into a \emph{physical}
          qubit structure (right) with limited connectivity. Node ``a'' is mapped to a chain of two nodes
          representing ``a'', which illustrates that the amount of physical qubits required to represent
          a problem is larger than the amount needed for a structural description.}\label{fig:embedding}
\end{figure}


\section{Quality Assessment of Quantum 3-SAT}\label{sec:ksat_on_dwave}
Let us now turn our attention to discussing how implementation details of 
quantum computational primitives can influence qualities of software. 
We focus on the problem of finding and comparing reductions of the problem to a machine
specific structure. Such low-level issues are usually not of much relevance in 
software engineering, and are justifiably perceived as implementation details---however,
we show that this level of abstraction is far from reached on quantum machines yet.

The \(k\)-SAT problem, the cornerstone of NP-completeness~\cite{karp1972reducibility}, serves as an example.
We first discuss different reductions of \(k\)-SAT \(\leq_{P}\) QUBO, and show how differences
arise from seemingly small details. We then offer guidance on comparing reductions. 

\subsection{Problem Definition}
\label{sec:3sat}
The problem of Boolean satisfiability is well known:
Let \(X := \{x_1, x_2, \dots, x_n\}\) be a set of Boolean variables, and let \emph{literals} be
defined as \(L := \{l | l \in \{x, \overline{x}\}, x \in X\}\). The set of all \emph{clauses} is given
by \(C := \{C_i | i \in [1; n], C_i \subset L, |C_i| = k\}\). For each \(x_i \in X\), there exists at least
one \(C_j\) such that \(x_i \in C_j\). A function \(f(\vec{x}) = \bigwedge_{C_i \in C} \bigvee_{l \in C_i} l\)
that satisfies these conditions is called a \(k\)-CNF function.   Given a \(k\)-CNF function \(f(\vec{x})\),
the \(k\)-SAT problem is to find an assignment \(\vec{x_t}\) such that \(f(\vec{x_t}) = \text{true}\). 
It is textbook knowledge~\cite{Sipser:2006} that every CNF formula can be cast in
3-CNF form.

 The \(k\)-SAT problem is the cornerstone of NP-completeness, but not all specific instances are difficult to solve.
The hardness of an instance depends on the ratio of clauses per variable \(\alpha = \frac{|C|}{|X|}\)~\cite{cheeseman1991really}. 
For instances with few clauses per variable (small \(\alpha\)), it is easy to find satisfying assignments. For instances with
many clauses per variable, it is easy to find contradictions. Instances with large or small values of \(\alpha\) tend to be
easy to solve. In an \(\alpha\) region surrounding \(\alpha_{c}\approx 4.25\), the probability that a random \(k\)-SAT formula can be satisfied
drops abruptly from 1 to 0~\citep{cheeseman1991really, mitchell1992hard}, and the hardest instances are contained in this
parameter regime. Improvements in SAT solving are therefore most desirable around this \emph{phase transition}.

\subsection{Choi's Standard Reduction}
Choi~\cite{Choi2010AdiabaticQA} gives a standard reduction from \(k\)-SAT to a QUBO. Let \(l_{ij}\) be the literal of variable \(x_{j}\) in clause \(i\).
Two literals \(l_{ij}\) and \(l_{i'j}\) are \emph{in conflict} if \(l_{ij} = \bar{l}_{i'j}\). Satisfying
a Boolean function in CNF implies at least one satisfied, but conflict-free literal per clause. 

The reduction assigns a negative weight \(-\omega\) to \(l_{ij}\): \(-\sum_{l_{ij}} \omega l_{ij}\). All
literals of a clause are fully connected with positive weighted edges: \(\sum_{l_{ij},l_{ij'} \in C_i} \delta l_{ij} l_{ij'}\).
All conflicting literals of the same variable are pairwise connected with \emph{positive} edge weights: 
\(\sum_{l_{ij} = \bar{l}_{i'j}} \delta l_{ij} l_{i'j}\). The last two sums are pure penalty terms and 
evaluate to 0 for correct assignments. This leads to a definition illustrated in Figure~\ref{fig:wmis_sat}:

\begin{definition}[\(k\)-SAT \(\leq_P\) QUBO (MIS)]\label{def:ksat_to_qubo}
	Let \(f(\vec{x})\) be a boolean \(k\)-CNF function. The literal of a variable \(x_j \in \vec{x}\) in a clause \(C_i \in C\) is given by \(l_{ij} \in \{0, 1\}\). Under the constraint \(\forall \delta, \omega : \delta > \omega > 0\), 
\begin{equation}
	\min[\vec{x}] \left( -\sum_{l_{ij}} \omega l_{ij} +
	\sum_{l_{ij},l_{ij'} \in C_i}
 \delta l_{ij} l_{ij'} +
	\sum_{l_{ij} = \overline{l_{i'j}}} \delta l_{ij} l_{i'j} \right)\label{eq:choi}
\end{equation}
	finds a satisfying assignment for \(f(\vec{x})\) if one exists.\citep{Choi2010AdiabaticQA}. 
\end{definition}

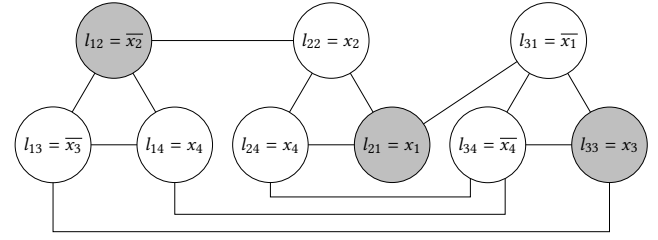
\begin{figure}
	\resizebox{\linewidth}{!}{\begin{tikzpicture}[scale=3]
\draw (0, 0) node[circle, fill=lightgray, draw] {
		\(l_{12} = \overline{x_2}\)
	  } -- 
	  ++(0.35, -0.6) node[circle, fill=white, draw] {
	  	\(l_{14} = x_4\)
	  } --
	  ++(-0.7, 0) node[circle, fill=white, draw] {
	  	\(l_{13} = \overline{x_3}\)
	  } -- (0, 0)
	  
	  (1.25, 0) node[circle, fill=white, draw] {
	  	\(l_{22} = x_2\)
	  } -- 
	  ++(0.35, -0.6) node[circle, fill=lightgray, draw] {
	  	\(l_{21} = x_1\)
	  } --
	  ++(-0.7, 0) node[circle, fill=white, draw] {
	  	\(l_{24} = x_4\)
	  } -- (1.25, 0)
	  
	  (2.5, 0) node[circle, fill=white, draw] {
	  	\(l_{31} = \overline{x_1}\)
	  } -- 
	  ++(0.35, -0.6) node[circle, fill=lightgray, draw] {
	  	\(l_{33} = x_3\)
	  } --
	  ++(-0.7, 0) node[circle, fill=white, draw] {
	  	\(l_{34} = \overline{x_4}\)
	  } -- (2.5, 0)
	  
	  (0, 0) -- (1.25, 0) ++(0.35, -0.6) -- (2.5, 0)
	  
	  (1.25, 0) ++(-0.35, -0.6) -- ++(0, -0.3) -- ++(1.15,0) -- ++(0, 0.3)
	  
	  (0,0) ++(0.35, -0.6) -- ++(0, -0.4) -- ++(1.9, 0) -- ++(0, 0.4)
	  
	  (0,0) ++(-0.35, -0.6) -- ++(0, -0.5) -- ++(3.2, 0) -- ++(0, 0.5);
\end{tikzpicture}}
	\caption{Graphical illustration of a QUBO formula that represents
    \(f(\vec{x}) = (\overline{x_2} \vee \overline{x_3} \vee x_4) \wedge (x_1 \vee x_2 \vee x_4 ) \wedge (\overline{x_1} \vee x_3 \vee \overline{x_4})\) using Choi's reduction. Grey nodes represent a satisfying assignment \([x_1 \mapsto 1, x_2 \mapsto 0, x_3 \mapsto 1]\).}
	\label{fig:wmis_sat}
\end{figure}

\subsection{Backbone Reduction}
To demonstrate the effect of different reductions on various aspects of QA performance,
consider a different reduction that we have devised for this paper, and that improves
(as we will analyse later) on the reduction given in Eq.~(\ref{eq:choi}):

\begin{definition}[\(k\)-SAT \(\leq_P\) QUBO (Backbone)]	\label{def:alt_sat_to_qubo}
    Let \(f(\vec{x})\) be a Boolean function in \(k\)-CNF, and let \(l_{ij}\) be a literal of \(x_j \in \vec{x}\) in \(C_i \in C\), with \(l_{ij}, x_j \in \{0,1\}\). Then
	\begin{equation}
	q(\vec{x}) = \omega \left( \sum_{l_{ij}, l_{ij'}} l_{ij} l_{ij'} +
	\sum_{l_{ij} = x_j} -l_{ij} x_j +
	\sum_{l_{ij} \neq x_j} -l_{ij} + l_{ij} x_j \right)
	\end{equation}
	with \(\omega > 0\) describes a QUBO \(q(\vec{x})\) for which \(\min[\vec{x}]\) represents
	a satisfying assignment of \(f(\vec{x})\) if one exists. 
\end{definition}

\noindent Mathematical details of the derivation are given in Appendix~\ref{sec:appendix}.

\subsection{Quality Criteria for Reductions}
Quality criteria for software are plentiful, and many of them also apply to the relative merits
of reductions. Since the development of quantum computers is mainly driven by the desire for
more computational power, we focus on two indicators: Performance and scalability.
There is (despite recent \href{https://quantum.ieee.org/images/files/pdf/ieee-support-for-standards.pdf}{standardisation efforts}) no universally 
applicable (and accepted) definition of how to measure performance of quantum computers;
this is particularly hard for QA, were the run-time is not determined by the input, but
chosen as a parameter---the annealing time. Consequently, we consider solution
\emph{quality}---how likely is it to obtain a correct answer that does not violate constraints,
and how accurate is the answer (\ie, how close is it to the optimal achievable value)---as proxy
for performance. Scalability considers the question of how large problems can be solved
on a hardware of given size (\ie, number of physical qubits). 

The achievable accuracy of a reduction depends on its structure (how well do logical connections
between qubits match the available physical structure?) and on hardware parameters. While the
adiabatic theorem ensures that longer annealing times (runtimes) results in better accuracy,
flaws and approximate implementations of the scheme in real hardware lead to less direct relations.
Like with traditional approximation algorithms, increasing the amount of computes samples
also leads to more accurate solutions. 


\subsection{Generating Instance Datasets}\label{sec:instance_datasets}

Owing to the lack of a published, physically accurate model of the quantum annealer that
includes imperfections and noise,\footnote{It is unlikely that such a physically accurate
model will be available in the near- or mid-term future.} determining scalability and
accuracy is currently only possible with experimental means~\cite{mcgeoch2019principles}.


When SAT is used to model constrained optimisation problems in practical applications, the
resulting SAT instances often exhibit specific structural properties, which can guide
the generation of useful test instances for determining quality properties of reductions. This
is, of course, not unlike the well-studied problem of generating tailored input data for
general software testing problems~\cite{Saswat:2013}.

We are interested in a general comparison of reductions, and therefore base our input data generation
on general properties of 3-SAT. We have discussed that the problem exhibits different regimes
in Section~\ref{sec:3sat}, and systematically generate random 3-SAT instances that cover these
by sweeping across different values of \(\alpha\). 
For the number of needed qubits, \(k|C|\) is the dominant term for both reduction approaches. Keeping \(|C|\) fixed and varying \(|V|\) produces stable QUBO sizes across the \(\alpha\)-spectrum, which guarantees a consistent hardware graph utilization.  




\subsection{Experimental Results}

\begin{figure}
	\includegraphics[width=\linewidth]{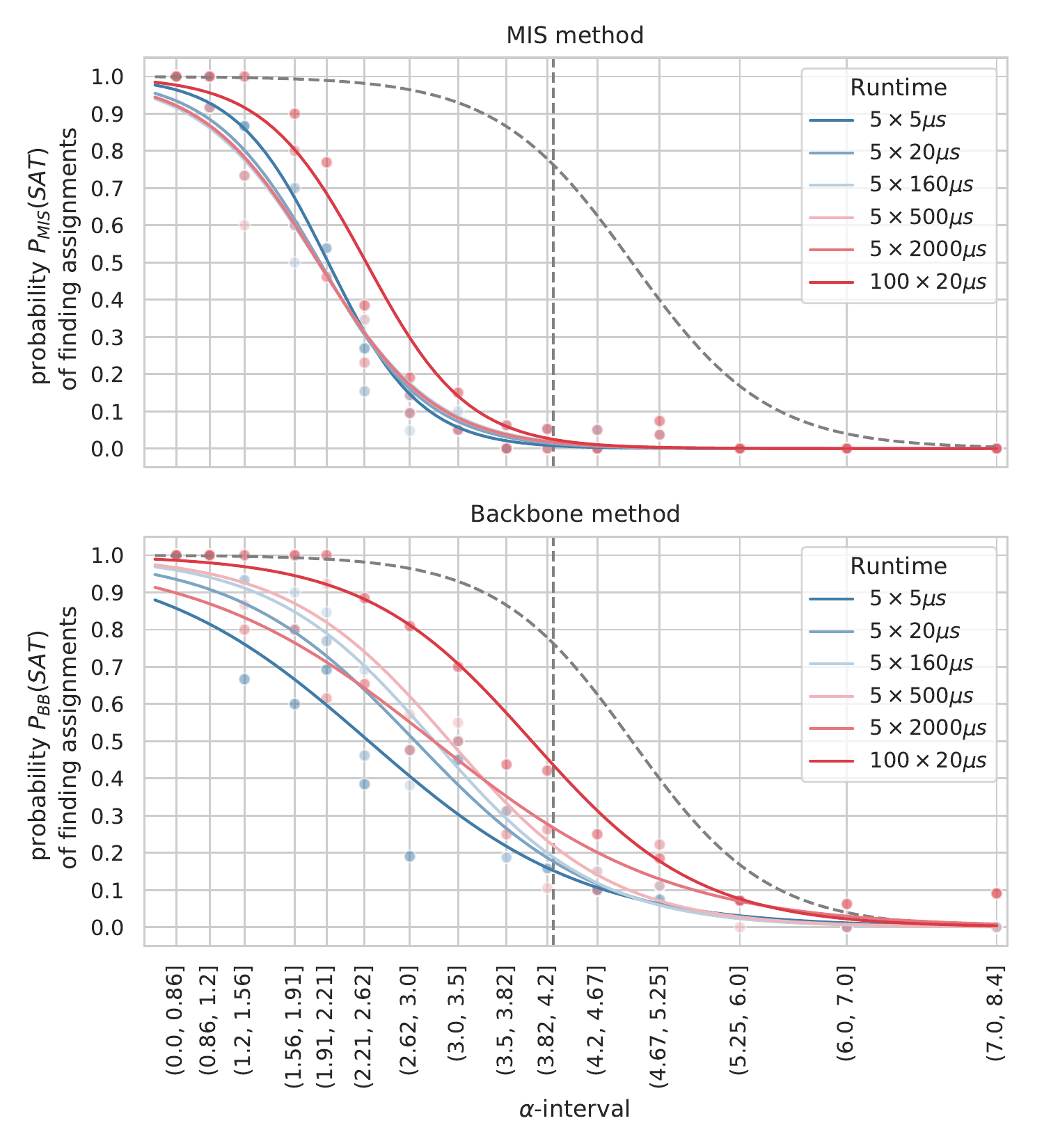}
	\caption{Influence of the embedding method on the probability of finding correct satisfying assignments for randomly
	generated 3-SAT instances with varying ratios \(\alpha\) of clauses to variables. The horizontal dashed lines marks the critical value \(\alpha_{c}\) (accompanied by a peak increase in required computing time when using traditional numeric
	solvers). The dashed curve represents the probability distribution of finding a satisfying assignment with optimal solvers. To ease comparing quantum and classical result, a logistic regression curve is given for each
    parameter variation.\newline
    Both (mathematically equivalent) methods arrive at correct conclusions less often than classical solvers, which is caused by imperfections and limitations of the available hardware.}\label{fig:prob_assignments}
\end{figure}

We generate a data-set containing 250 random 3-SAT instances with 42 clauses each. In
total, six runs with varying annealing times (5 to 2000\(\mu\)s) and samples sizes 
(5 and 100) were performed on the quantum annealer described in Section~\ref{sec:dwave_platform}.
Figure~\ref{fig:prob_assignments} shows results for the two different reduction methods. 

\subsubsection{Accuracy}
For Choi's MIS-based reduction, the annealing time does not substantially  effect the
accuracy. Using probability amplification by performing a larger number of runs \(100 \times 20 \mu s\),
does improve the accuracy slightly. Results obtained with the backbone method, in contrast,
improve with increasing annealing time, and increasing the number of runs is also accompanied by 
a larger improvement as compared to the MIS method. It is also important to note that using an excessively
long annealing time of 2000\(\mu\)s results in a \emph{decrease} of result quality\footnote{This observation does
contradict the adiabatic theorem; the effect is likely caused by a large amount 
of noise leading incorrect initial configurations or random energy level jumps during the annealing
process. Both effects occur with growing probability for increasing annealing times.}.

Recall from the above discussion that solving SAT instances in the critical parameter region around
\(\alpha_{c}\) is most involved for classical solvers, and improvements by quantum computers are most
desirable in this region. Unfortunately, the MIS method delivers satisfying solutions in this range with
almost zero probability.

\begin{figure}
	\includegraphics[width=\linewidth]{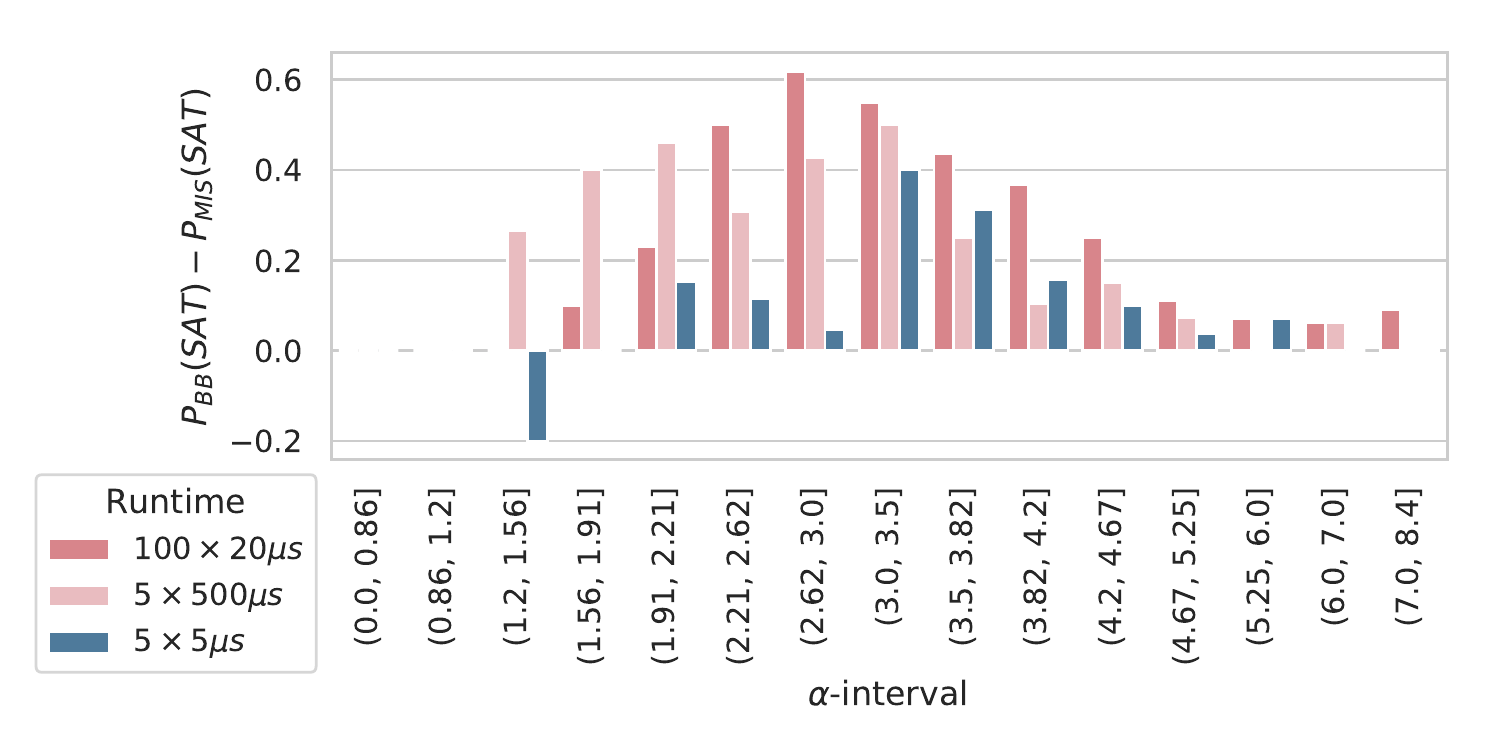}
	\caption{Accuracy difference (in percentage points) between MIS and backbone method.}
	\label{fig:improvement}
\end{figure}

Figure~\ref{fig:prob_assignments} directly compares accuracy results. The difference in accuracy reaches
up to 60\%, and the backbone method is consistently more accurate for all scenarios. The decreasing difference
in accuracy at  \(\alpha > \alpha_{c}\) is a consequence of the low number of satisfiable instances in
this region. Around the critical region, we observe marked differences of around 35\%.

\begin{figure}
	\includegraphics[width=\linewidth]{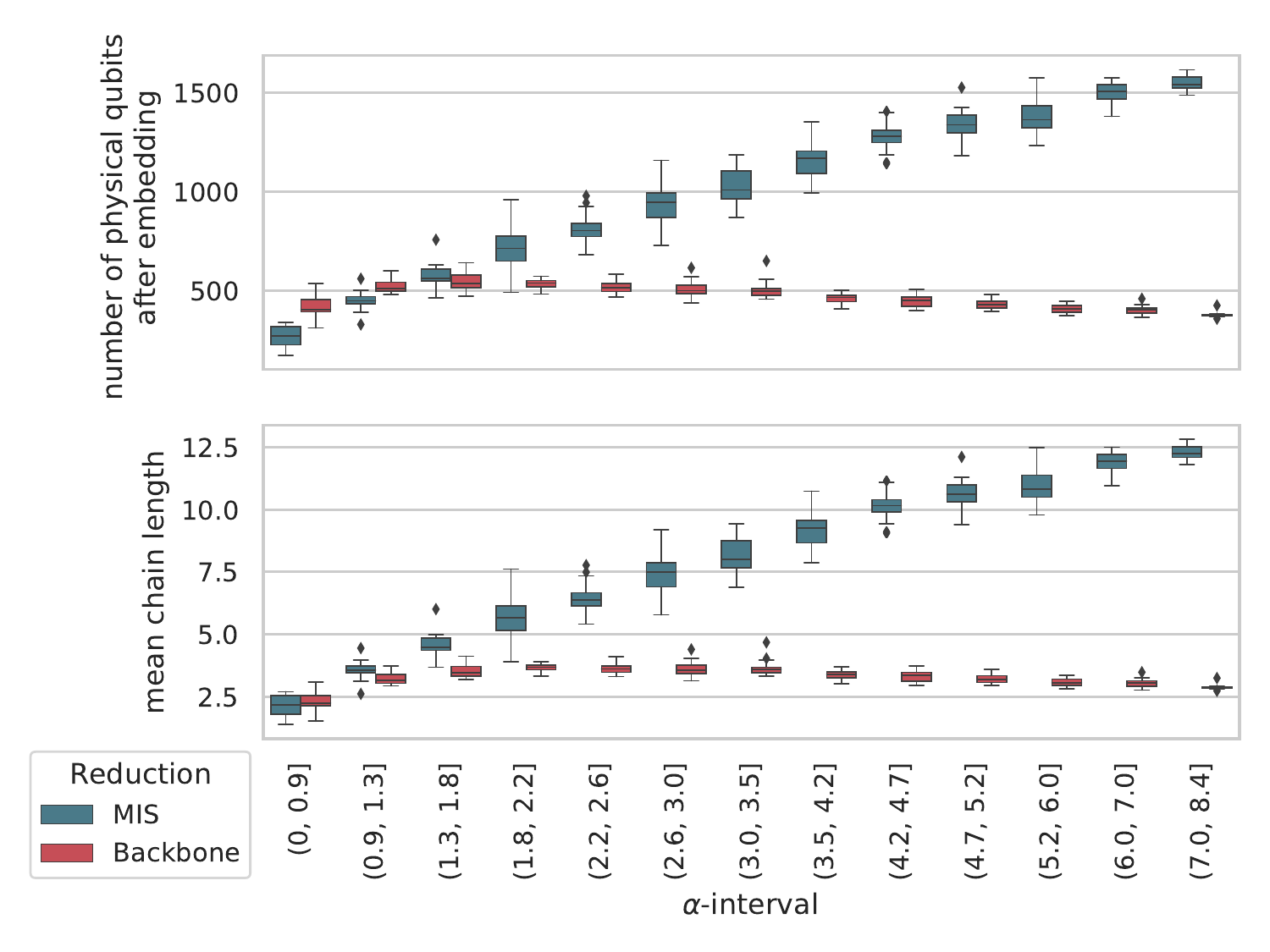}
	\caption{Number of required physical qubits to after embedding a QUBO for a given ratio of
	 variables and clauses (top) and median chain lengths necessary to connect qubits without
	 direct interconnections (bottom).}	\label{fig:embedding_quality}
\end{figure}

\subsubsection{Scalability}
Figure~\ref{fig:embedding_quality} compares scalability of the two reductions by analysing the amount
of required physical qubits, and the mean length of chains necessary to connect qubits without
direct physical connections (we use the \emph{minorminer} tool provided as part of the D-Wave
API to embed QUBOs into the hardware graph).

For the MIS-based method, the amount of physical qubits and mean chain length grow
essentially linear with an increasing \(\alpha\), which follows from the
pairwise links between conflicting literals. 

The backbone method improves upon both aspects because the QUBO is less densely populated,
which makes it easier to find embeddings. Especially around the critical value \(\alpha_{c}\),
the amount of required physical qubits is only half of what is required for the MIS method,
which in turn implies that substantially larger problem sets can be solved on a hardware of
given size.

\section{Conclusion}\label{sec:conclusion}
Development and evaluation of quantum software components must address
well-established engineering concerns of traditional SWE.
Based on the scenario of replacing SAT solving, a key element of many applications,
with a quantum component, we have shown that careful attention is required
in defining and evaluating relevant qualities. We have argued that scalability
and accuracy are of particular relevance for early existing 
quantum annealers. While replacing classical with quantum components is not
particularly involved from a programming perspective, our experiments
indicate that engineers must be aware of crucial details that might be perceived as irrelevant in traditional SWE to make informed decisions on potentials and pitfalls of quantum computing.



\bibliographystyle{ACM-Reference-Format}
\bibliography{ms}

\begin{appendix}
\section{Alternative Reductions}\label{sec:appendix}
\textbf{Loosened Clause Penalties}
Choi's reduction (\cref{def:ksat_to_qubo}), is, in essence, a reduction from \(k\)-SAT to the problem of
finding the maximal independent set (MIS) of a given graph. Assume a \(k\)-SAT instance is reduced to QUBO as described in \cref{def:ksat_to_qubo}, and let \(G_f\) be the graph representation.
Consider a MIS of \(G_f\), which is given by the largest set of vertices such that there are no
connected vertices. In \cref{def:ksat_to_qubo} this
property is enforced by the constraint \(\delta > \omega\). Solving a QUBO defined by \cref{def:ksat_to_qubo} also solves the MIS problem for \(G_f\). The problem of finding the MIS \(G_f\) corresponds to the problem of finding the maximal number of satisfiable clauses (MAX-k-SAT) in \(f\). The relation between MAX-k-SAT and k-SAT is trivial.

\begin{theorem}
	\label{thm:loosen}
	Setting \(\delta = \omega\) in \cref{def:ksat_to_qubo} does not change the correctness of the assignment derived from the QUBO solution.
\end{theorem}

\begin{proof}\label{prf:loosen}
	Let \(q(x)\) be a sub-QUBO representing one clause like described in \cref{def:ksat_to_qubo}. Under \(\delta = \omega\) the following holds: \(\min(E(n)) = -\omega\) if \(E(n)\) is the energy of a clause sub-QUBO with \(n\) satisfied literals.
	It is straight forward to see that \(E(n) = -n\omega + \binom{n}{2}\omega\). Therefor, \(E(0) = 0\) and \(E(1) = E(2) = -\omega\). The inequality \(E(n) < E(n+1)\) evaluates to \(-n < -1\) which is true for all \(n > 1\). This leads to the conclusion that \(\min(E(n)) = -\omega\).
	
	Consider a clause \(C_i\) and its corresponding sub-QUBO \(q_i(x)\). Now, \(\min(q_i(x)) = -\omega\) for one or two satisfied literals in \(q_i(x)\). Therefore, the minimization of \(q_i(x)\) leads to a satisfied clause \(C_i\). For two conflict-free clauses \(C_i\) and \(C_j\) the combined minimum energy is given by \(\min(q_i(x) + q_j(x)) = -2\omega\). Now we introduce a conflict between \(C_i\) and \(C_j\). That activates an additional penalty term \(p_{ij} = \omega\) which leads to \(\min(q_i(x) + q_j(x) + p_{ij}) = -\omega > -2\omega\). This shows that conflicts between clauses always lead to a higher energy level and thus should be avoided when minimizing the complete QUBO \(q(x)\). For all satisfiable k-SAT instances \(f(x)\) with \(n\) clauses the minimal energy of their corresponding QUBOs \(q(x)\) will be \(min(q(x)) = -n\omega\). Every function \(f(x)\) with  minimal QUBO value \(\min(q(x)) > -n\omega\) cannot	be satisfied.
\end{proof}


\textbf{Backbone}
Choi's reduction represents variables solely by their literals. To avoid conflicts, we need to ensure that \(l_{ij} \neq l_{i'j}\) for all pairs \((l_{ij}, l_{i'j}) \Leftrightarrow (x_j, \bar{x}_j)\). An edge in the QUBO connects the
literals as penalty term, which leads to highly connected graphs for instances with large values of
\(\alpha\). The degree of connectivity can be improved by introducing a \emph{backbone} for
variables, which allows us to transitively express equivalence between literals 
by linking them to their corresponding variable. The reduction is given in
Definition~\ref{def:alt_sat_to_qubo} on page~\pageref{def:alt_sat_to_qubo}.

\begin{proof} [Correctness of \cref{def:alt_sat_to_qubo}]
    Let \(E(n)\) be the energy of a clause sub-QUBO with \(n\) satisfied literals. The
	difference between sub-QUBOs in \cref{def:ksat_to_qubo} and \cref{def:alt_sat_to_qubo} is
	that in the latter node weights of literals \(l_{ij} \neq x_j\) are moved to the edges \((l_{ij},	x_j)\). For every literal, there exists exactly one edge to its corresponding variable.
	Therefore, edge weights can be viewed as node weights, and it follows that \(E(n) = -n\omega +
	\binom{n}{2}\omega\). Consequently, \(\min(E(n)) = -\omega\) also holds for
	\cref{def:alt_sat_to_qubo}. If two literals \(l_{ij}\) and \(l_{i'j}\) conflict, one of
	 \(\omega l_{ij} x_j \) or \(\omega l_{i'j} x_j\) evaluates to
	\(\omega\), while the other evaluates to 0. The rest of the argument follows Theorem~\ref{thm:loosen}.
\end{proof}

\end{appendix}

\end{document}